\documentclass[reqno]{amsart}
\usepackage{amssymb,epsfig,graphics,graphicx,bbm,hyperref,color}
\usepackage{multicol}
\usepackage{enumitem}
\usepackage{tikz-cd}
\usetikzlibrary{matrix,arrows,decorations.pathmorphing,positioning}

\definecolor{gray}{rgb}{0.93,0.93,0.93}
\definecolor{light-gold}{rgb}{0.99,0.97,0.78}

\def\be{\begin{equation}}
\def\ee{\end{equation}}
\def\bm{\begin{multline}}
\def\bfig{\begin{figure}[htb]}
\def\efig{\end{figure}}
\newcommand{\dd}{{\rm d}}
\newcommand{\e}[1]{\,{\rm e}^{#1}\,}

\newcommand{\sumtwo}[2]{\sum_{\substack{#1 \\ #2}}}
\newcommand{\sumthree}[3]{\sum_{\substack{#1 \\ #2 \\ #3}}}

\newcommand{\isdefby}{\;\doteqdot\;}

\numberwithin{equation}{section}
\newtheorem{theorem}{Theorem}[section]
\newtheorem{proposition}[theorem]{Proposition}

\newcommand{\eps}{{\varepsilon}}

\newcommand{\caE}{{\mathcal E}}
\newcommand{\caG}{{\mathcal G}}

\newcommand{\caM}{{\mathcal M}}
\newcommand{\caP}{{\mathcal P}}

\newcommand{\caV}{{\mathcal V}}
\newcommand{\caW}{{\mathcal W}}
\newcommand{\caX}{{\mathcal X}}

\newcommand{\bbE}{{\mathbb E}}
\newcommand{\bbN}{{\mathbb N}}
\newcommand{\bbP}{{\mathbb P}}
\newcommand{\bbR}{{\mathbb R}}
\newcommand{\bbS}{{\mathbb S}}
\newcommand{\bbZ}{{\mathbb Z}}
\newcommand{\bsone}{{\boldsymbol 1}}
\newcommand{\bsJ}{{\boldsymbol J}}

\newcommand{\bsm}{{\boldsymbol m}}
\newcommand{\bsq}{{\boldsymbol q}}

\newcommand{\bsx}{{\boldsymbol x}}
\newcommand{\bsw}{{\boldsymbol w}}
\newcommand{\bspi}{{\boldsymbol\pi}}
\newcommand{\bsphi}{{\boldsymbol\phi}}
\newcommand{\bsvarphi}{{\boldsymbol\varphi}}

\makeatletter
  \def\tagform@#1{\maketag@@@{\scriptsize{(#1)}\@@italiccorr}}
\makeatother

\renewcommand{\eqref}[1]{(\ref{#1})}

\begin{document}

%{\hfill\small \version} \vspace{2mm}

\title{Loop correlations in random wire models}

\author{Costanza Benassi}
\address{Department of Mathematics, Physics \& Electrical Engineering, Northumbria University,
Newcastle upon Tyne, NE1 8ST, United Kingdom}
\email{costanza.benassi@northumbria.ac.uk}

\author{Daniel Ueltschi}
\address{Department of Mathematics, University of Warwick,
Coventry, CV4 7AL, United Kingdom}
\email{daniel@ueltschi.org}



\begin{abstract}
We introduce a family of loop soup models on the hypercubic lattice. The models involve links on the edges, and random pairings of the link endpoints on the sites. We conjecture that loop correlations of distant points are given by Poisson-Dirichlet correlations in dimensions three and higher. We prove that, in a specific random wire model that is related to the classical XY spin system, the probability that distant sites form an even partition is given by the Poisson-Dirichlet counterpart.
\end{abstract}

\thanks{\copyright{} 2018 by the authors. This paper may be reproduced, in its
entirety, for non-commercial purposes.}

\maketitle

{\small\tableofcontents}

\section{Introduction}
\label{sec intro}

Loop soups are models in statistical mechanics that involve sets of one-dimensional loops living in higher dimensional space. These models are representations of particle or spin systems of statistical physics. It was recently conjectured that in most cases, in dimensions three and higher, these models have phases with long, macroscopic loops --- the lengths of these loops scale like the volume of the system --- and the joint distribution of macroscopic loops is always Poisson-Dirichlet  \cite{GUW}.

This conjecture has been rigorously established in a model of spatial permutations related to the quantum Bose gas \cite{BU,BZ,EP}. This model has a peculiar structure that makes it possible to integrate out the spatial variables and to use tools from asymptotic analysis, so there were suspicions that this property was accidental. But the conjecture has also been verified numerically in several other models, namely in lattice permutations \cite{GLU}; in loop $O(N)$ models \cite{NCSOS}; and in the random interchange and a closely related loop model \cite{BBBU}. These findings are alas not supported by rigorous results.

There exist limited results for some models with fundamental spatial structure. The method of reflection positivity and infrared bounds \cite{FSS,DLS} allows to prove the occurrence of macroscopic loops \cite{Uel1}. More precisely, it is shown that the expectation of the length of a loop attached to a given site, when divided by the volume, is bounded away from 0 uniformly in the size of the system. While encouraging, this result gives no information regarding the possible presence of several macroscopic loops, let alone their joint distribution.

The goal of this article is to propose a genuinely spatial loop model where much of the conjecture can be rigorously established. We refer to it as the ``random wire model". It is defined for arbitrary finite graphs; in the most relevant case, the set of vertices is a large box in $\bbZ^d$ and the set of edges are the pairs of nearest-neighbours. On each edge there is a random number of ``links" satisfying the constraint that the number of links touching a site is even. These links are paired at each site, resulting in closed trajectories (an illustration can be found in Fig.\ \ref{fig loops}). Our main result is a rigorous proof that {\it even loop correlations} are given by Poisson-Dirichlet, at least when the parameters of the model are chosen wisely.

There is a lot of background for this study. Our random wire model is an extension of the random current representation of the Ising model that was introduced by Griffiths, Hurst, and Sherman \cite{GHS}, and popularised by Aizenman \cite{Aiz}. It is also related to the Brydges-Fr\"ohlich-Spencer representation of spin $O(N)$ models \cite{BFS,FFS}, and to loop $O(N)$ models \cite{PS,BST}. The Poisson-Dirichlet distribution of random partitions was introduced by Kingman \cite{Kin}; it is the invariant measure for the split-merge (coagulation-fragmentation) process \cite{Tsi,DMZZ,Ber}. Its relevance for mean-field loop soup models was first suggested by Aldous for the random interchange model on the complete graph, see \cite{BD}; Schramm succeeded in making this rigorous \cite{Sch} (see also \cite{BK,Bjo1,Bjo2,BKLM}). The relevance of these ideas for systems with spatial structure was pointed out in \cite{GUW}. 

The connections between the Poisson-Dirichlet distribution and symmetry breaking was noticed and exploited in \cite{Uel1,NCSOS}. Our method of proof combines these ideas and rests on two major results about the classical XY model: The proof of Fr\"ohlich, Simon, and Spencer that a phase transition occurs in dimensions three and higher \cite{FSS}; and Pfister's characterisation of all translation-invariant extremal infinite-volume Gibbs states \cite{Pfi2}. We should point out that the precise relations between Poisson-Dirichlet and symmetry breaking are far from elucidated. The heuristics of Section \ref{sec heuristics} show that the loops that represent the classical XY model are characterised by the distribution PD(1), as are the loops of the quantum XY model \cite{Uel1}. However, these heuristics also show that the loops representing the classical Heisenberg model are characterised by PD($\frac32$) while the loops of the quantum Heisenberg model are PD(2) \cite{GUW,Uel1}. Right now, this looks curious.

It is perhaps worth emphasising that the Poisson-Dirichlet distribution is expected to characterise loop soups only in dimensions three and higher. The behaviour in dimension two is also interesting and partially understood, see \cite{Betz,BFU,EP}. There may be a Berezinskii--Kosterlitz--Thouless phase where loop correlations have power-law decay instead of exponential. A separate topic is the critical behaviour of two-dimensional loop soups, that is characterised by conformal invariance and Schramm-L\"owner evolution; there have been many impressive results in recent years, but we do not discuss this here.

The article is organised as follows. The notation is summarised in Section \ref{sec not} for the comfort of the reader. The random wire model is introduced in Section \ref{sec setting} and basic properties are established. The Poisson-Dirichlet conjecture is explained in Section \ref{sec PD}. Our main results, Theorems \ref{thm long dens} and \ref{thm PD}, are stated in Section \ref{sec results}. The first claim deals with the density of points in long loops, and the second claim is about even loop correlations being given by Poisson-Dirichlet. Section \ref{sec wire O(N)} discusses classical spin systems and their relations to the random wire model. We gather the necessary properties in Section \ref{sec proofs} by summarising and completing the results of \cite{FSS,Pfi2}, and we prove Theorems \ref{thm long dens} and \ref{thm PD}.

\section{Notation}
\label{sec not}
We list here the main notation used in this article; the precise definitions can be found in subsequent sections.

\begin{itemize}
\item $\caG = (\caV,\caE)$ the graph; $\caV$ is the set of vertices and $\caE$ is the set edges. $\caG^{\rm b} = (\caV \cup \bar\caV, \caE \cup \bar\caE)$ denotes the graph with a boundary; $\bar\caV$ is the set of boundary sites and $\bar\caE$ are edges between $\caV$ and $\bar\caV$.

\item $\caG_L = (\Lambda_L,\caE_L)$ with $\Lambda_L = \{-L,\dots,L\}^d \subset \bbZ^d$ and $\caE_L$ the set of nearest-neighbours. $\caG_L^{\rm b}$ is the graph with boundary $\partial\Lambda_L$, given by sites of $\bbZ^d$ at distance 1 from $\Lambda_L$.

\item $\caW_\caG = \{ \bsw = (\bsm,\bspi) \}$ is the set of wire configurations on $\caG$, that consists of a link configuration $\bsm \in \caM_\caG \subset \bbN_0^\caE$ (with an even number of links touching each site) and a pairing configuration $\bspi \in \caP_\caG(\bsm)$.

\item $n_x(\bsm)$ is the local occupancy (or ``local time"); it is equal to the number of times that loops pass by the site $x \in \caV$.

\item $\lambda(\bsw)$ is the number of loops in the wire configuration $\bsw$.

%\item $\Gamma_\caG$ is the set of loops on $\caG$ and $\caR_\caG = \bbN_0^{\Gamma_\caG}$ is the set of loop configurations on $\caG$.

%\item $L : \caC_\caG \to \caR_\caG$ is the map that returns the loops of a current configuration. $|L(\bsm,\bspi)|$ is the number of loops.

\item $\alpha$ is the positive ``loop parameter".

\item $\bsJ = (J_e)_{e \in \caE}$ are ``edge constants", or ``coupling parameters".

\item $U : \bbN_0 \to \bbR$ is a potential function; $U(n_x)$ gives the energy of the $n_x$ wires that cross the site $x\in\caV$.

\item $\bbP_\caG^{\alpha,\bsJ}, \bbE_\caG^{\alpha,\bsJ}$ denote the probability and expectation with respect to wire configurations.

\item $Z_\caG(\alpha,\bsJ)$ is the partition function and $p_\caG(\alpha,\bsJ)$ is the pressure.

\item $\tilde n_x$ is the number of pairs at the site $x$ that belong to long or open loops.

\item $E_X(\bsx,\bsq)$ is the set of configurations $\bsw$ where $(x_i,q_i)$ and $(x_j,q_j)$ belong to the same loop iff $i,j$ belong to the same partition element of $X$.

\item $\caX_{2k}^{\rm even}$ is the family of set partitions of $\{1,\dots,2k\}$ whose elements have even cardinality.

\item $M_\theta(X)$ is the probability that $k$ random points on $[0,1]$ and a random partition chosen with Poisson-Dirichlet distribution PD($\theta$), yield the set partition $X$.

\item $M_\theta^{\rm even}(2k) = \sum_{X \in \caX_{2k}^{\rm even}} M_\theta(X)$ is the probability that $2k$ random points on $[0,1]$, and a random partition from PD($\theta$), yield an even set partition.
\end{itemize}

\section{Setting}
\label{sec setting}

\subsection{Links, pairings, wires, and loops}
\label{sec links}

We consider a generalisation of the model of random currents of the Ising model. Let $\caG = (\caV,\caE)$ a graph. Given a collection $\bsm = (m_e)_{e\in\caE}$ of nonnegative integers, we define the local occupancy (or local time) to be
\be
\label{def local occupancy}
n_x(\bsm) \isdefby \tfrac12 \sum_{e \in \caE, e \ni x} m_e.
\ee
A {\bf link configuration} is a collection $\bsm$ that satisfies the constraint that there is an {\it even} number of links touching any site; in other words, the local occupancy $n_x(\bsm)$ is integer at every site $x \in \caV$. We let $\caM_\caG$ denote the set of link configurations. A link configuration can be represented by a labeled multigraph with labeled edges, see Fig. \ref{fig multigraph}.

\bfig
\includegraphics[width=55mm]{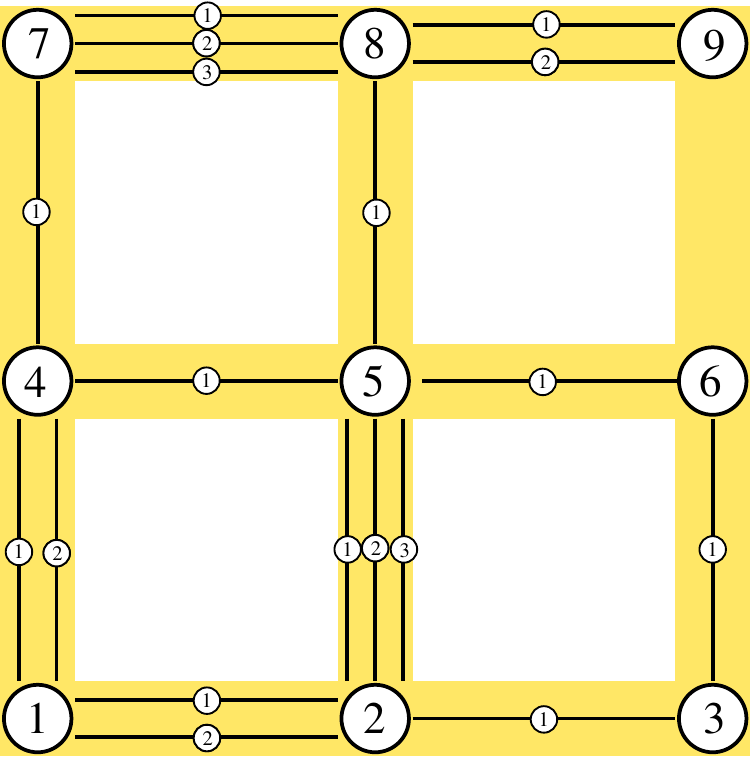}
\caption{A link configuration is represented by a labeled multigraph. Here the graph is the square lattice $\{1,2,3\}^2$ and edges are nearest-neighbours.}
\label{fig multigraph}
\efig

For a given link configuration $\bsm$, a {\bf pairing configuration} $\bspi = (\pi_x)_{x \in \caV}$ is a collection of pairings such that $\pi_x$ connects the links that touch the site $x \in \caV$. This is illustrated in Fig.\ \ref{fig loops}. We let $\caP_\caG(\bsm)$ denote the set of pairing configurations that are compatible with $\bsm$; notice that the number of pairing configurations is equal to
\be
\label{number pairings}
|\caP_\caG(\bsm)| = \prod_{x \in \caV} \bigl( 2n_x(\bsm)-1 \bigr) !!
\ee
(with the convention that $(-1)!! = 1$).
We call the pair $\bsw = (\bsm,\bspi)$ a {\bf wire configuration}; the set of wire configurations on the graph $\caG$ is denoted $\caW_\caG$.

We now define the loops of a wire configuration. This notion is intuitive and it is illustrated in Fig.\ 2 (b), even though the proper definition is a bit cumbersome. We consider the set of finite sequences of labeled links $\bigl( (e_1,p_1), \dots, (e_\ell,p_\ell) \bigr)$ where $e_i \in \caE$ and $p_i \in \{1,\dots,m_{e_i}\}$, and such that $e_i \cap e_{i+1} \neq \emptyset$, $i = 1,\dots,\ell$. We identify sequences that are related by cyclicity and inversion; that is, we identify $\bigl( (e_2,p_2),\dots,(e_\ell,p_\ell), (e_1,p_1) \bigr)$ and $\bigl( (e_\ell,p_\ell),\dots,(e_1,p_1) \bigr)$ with $\bigl( (e_1,p_1),\dots,(e_\ell,p_\ell) \bigr)$. After identification, these sequences form a {\bf loop} of length $\ell$. In order to define the set of loops of a given wire configuration $\bsw$, we can start at any link $(e_1,p_1)$; we choose an endpoint $x$ and get the next link as the one that is paired by the pairing $\pi_x$; we continue until we get back to $(e_1,p_1)$. For the next loop we choose a link that has not been selected yet, and we proceed alike until all links have been exhausted.

The number of loops of a wire configuration $\bsw$ is denoted $\lambda(\bsw)$. We also define the length of a loop as the number of links in the loop.

\bfig
\includegraphics[width=56mm]{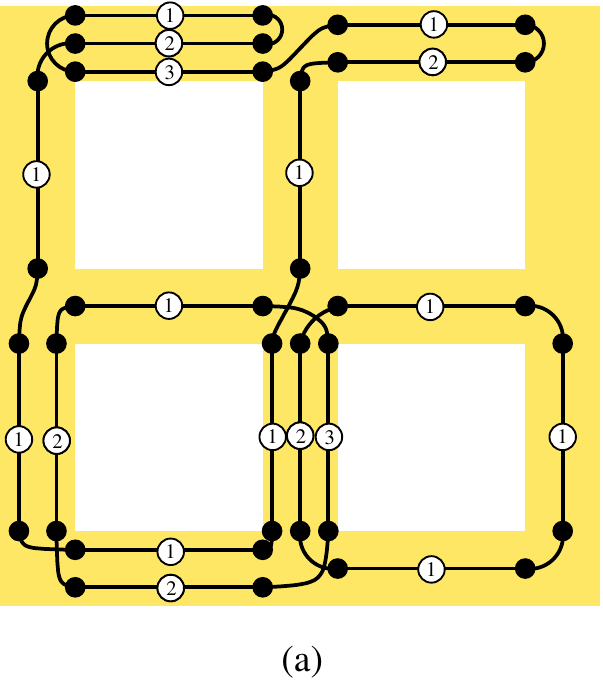} \qquad\qquad
\includegraphics[width=56mm]{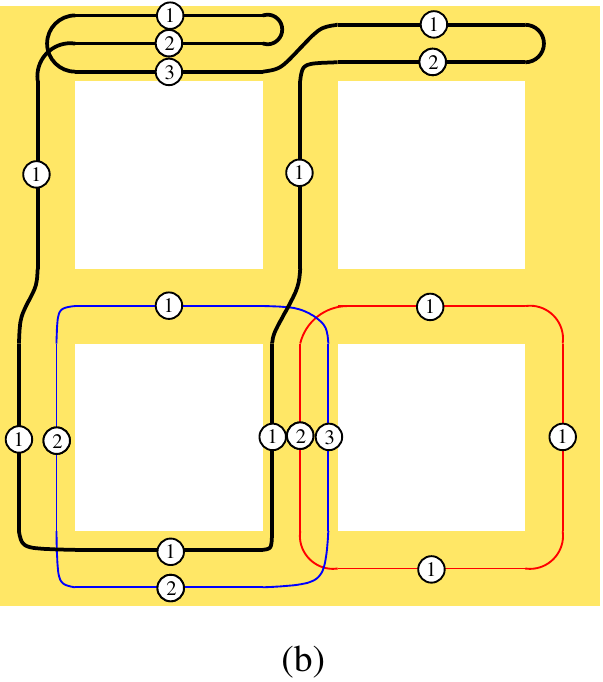}
\caption{(a) A wire configuration consists of links and pairings. (b) It gives rise to a loop configuration, here with three loops.}
\label{fig loops}
\efig

\subsection{The model of random wires}

We now introduce the probability distribution on wire configurations.
Let $\bsJ = (J_e)_{e \in \caE}$ a collection of nonnegative parameters indexed by the edges of $\caG$. Let $\alpha>0$ another parameter. We consider an ``interaction potential" function $U : \bbN_0 \to \bbR \cup \{+\infty\}$ and define the probability of the wire configuration $\bsw = (\bsm,\bspi)$ to be
\be
\label{def proba}
\bbP_\caG^{\alpha,\bsJ}(\bsw) = \frac1{Z_\caG(\alpha,\bsJ)} \alpha^{\lambda(\bsw)} \biggl( \prod_{e \in \caE} \frac{J_e^{m_e}}{m_e!} \biggr) \exp\biggl\{ -\sum_{x\in\caV} U \bigl( n_x(\bsm) \bigr) \biggr\}.
\ee
Here, the normalisation $Z_\caG(\alpha,\bsJ)$ is the partition function defined by
\be
Z_\caG(\alpha,\bsJ) = \sum_{\bsw \in \caW_\caG} \alpha^{\lambda(\bsw)} \biggl( \prod_{e \in \caE} \frac{J_e^{m_e}}{m_e!} \biggr) \exp\biggl\{ -\sum_{x\in\caV} U \bigl( n_x(\bsm) \bigr) \biggr\}.
\ee

Notice that the exponent of $\alpha$ is $\lambda(\bsw)$, which is the number of loops. For $\alpha\neq1$, loops affect the probability distribution.

The interaction potential typically becomes infinite as the local occupancy diverges. It is natural to consider models where the partition function is finite for all choices of $\alpha$ and $\bsJ$. The first claim of the next proposition gives a sufficient condition.

\begin{proposition}
\label{prop basic prop}
Let $\bar\alpha = \max(\sqrt\alpha,1)$ and assume that the potential function satisfies
\[
(2n-1)!! \e{-U(n)} \leq C^n
\]
for some positive constant $C$ independent of $n$. Then
\begin{itemize}
\item[(a)]
The partition function is bounded by
\[
Z_\caG(\alpha, \bsJ) \leq \exp\biggl\{ \bar\alpha C \sum_{e\in\caE} J_e \biggr\}.
\]
\item[(b)] Let $\ell^{\rm max}_{x_0}(\bsw)$ be the length of the longest loop that passes through the site $x_0$. For all $n\in\bbN$ and all $\eta\geq0$, we have
\[
\bbP_\caG^{\alpha,\bsJ}(\ell^{\max}_{x_0} \geq n) \leq \e{-\eta n} \sum_{k\geq0} \sumtwo{x_1,\dots,x_k \in \caV}{\{x_{i-1},x_i\} \in \caE \text{ for } i = 1,\dots,k-1} \prod_{i=1}^k \bigl( \e{\e\eta \bar\alpha C J_{\{x_{i-1},x_i\}}} - 1 \bigr)^{1/2}.
\]
\end{itemize}
\end{proposition}

The upper bound in (b) involves a sum over walks of arbitrary lengths that start at $x_0$.
In many situations, such as graphs with bounded degrees and $J_e$ bounded uniformly, this sum is convergent when $\bsJ$ is small. Then all loops passing by the site $x_0$ are small, that is, their lengths are finite uniformly in the size of the graph.

\begin{proof}
The number of loops is less than $\frac12 \sum_e m_e$ so that $\alpha^{\lambda(\bsw)} \leq \bar\alpha^{\sum m_e}$. The number of pairing configurations is $\prod_x (2n_x-1)!!$. Neglecting the constraints on link numbers, we get
\be
\begin{split}
Z_\caG(\alpha,\bsJ) &\leq \sum_{(m_e)_{e\in\caE}} \bar\alpha^{\sum_e m_e} \biggl( \prod_{e\in\caE} \frac{J_e^{m_e}}{m_e!} \biggr) C^{\sum_x n_x(\bsm)} \\
&= \e{\bar\alpha C \sum_e J_e}.
\end{split}
\ee
We used $\sum_{x \in \caV} n_x(\bsm) = \sum_{e \in \caE} m_e$, which follows from Eq.\ \eqref{def local occupancy}. This proves (a).

For the claim (b), given a configuration $\bsm$, we consider the graph with set of vertices $\caV$ and with set of edges $\{ e \in \caE : m_e \geq1 \}$. Further, let $\caG' = (\caV',\caE') \subset \caG$ be the connected subgraph that contains the vertex $x_0$. We have
\be
\bbP_\caG^{\alpha,\bsJ} \bigl( \ell_{\rm max}(x_0) \geq n \bigr) \leq \bbP_\caG^{\alpha,\bsJ} \Bigl( \sum_{y\in\caE'} n_y \geq n \Bigr) \leq \e{-\eta n} \bbE_\caG^{\alpha,\bsJ} \Bigr[ \e{\eta \sum_{y\in\caE'} n_y} \Bigr].
\ee
The last bound follows from Markov's inequality. We now condition on the graph $\caG'$ in the equation below; the sum over $(\bsm,\bspi) : \caG'$ is a sum over link configurations on $\caE'$ so that the graph with edges $\{ e \in \caE' : m_e \geq 1 \}$ is connected, and over pairing configurations on $\caV'$. Thanks to factorisation properties we have
\bm
\bbE_\caG^{\alpha,\bsJ} \Bigr[ \e{\eta \sum_{y\in\caE'} n_y} \Bigr] = \frac1{Z_\caG(\alpha,\bsJ)} \sum_{\caG'} Z_{\caG \setminus \caG'}(\alpha,\bsJ) \\
\sum_{(\bsm,\bspi) \, : \, \caG'} \alpha^{\lambda(\bsw)} \Bigl( \prod_{e \in \caE'} \frac{J_e^{m_e}}{m_e!} \Bigr) \e{-\sum_{y \in \caV'} (U(n_y)-\eta n_y)}.
\end{multline}
Assuming that $U$ is normalised so that $U(0)=0$, which we can do without loss of generality, we have that $Z_{\caG \setminus \caG'}(\alpha,\bsJ) \leq Z_\caG(\alpha,\bsJ)$. Using similar estimates as in (a), we get
\be
\bbE_\caG^{\alpha,\bsJ} \Bigr[ \e{\eta \sum_{y\in\caE'} n_y} \Bigr] \leq \sum_{\caG'} \sum_{\bsm:\caG'} \prod_{e \in \caE'} \frac{(\e\eta \bar\alpha C J_e)^{m_e}}{m_e!}.
\ee
For any connected graph, there exists a walk that uses each edge exactly twice. This is easily seen by induction: knowing the walk for a given graph, and adding an edge, we get a new walk by crossing the new edge twice. The sum over connected graphs $\caG'$ can then be estimated by a sum over walks starting at the vertex $x_0$. The sum over $\bsm$ can be estimated by $\e{\e\eta \bar\alpha C J_e}-1$ at every edge, and we get the claim (b).
\end{proof}

We now introduce a random wire model with ``open" boundary conditions. The idea is to allow open loops that end at the boundary (we refer to them as open loops, although they are no real loops). The new graph is 
\be
\label{eq open graph}
\caG^{\rm b} = \bigl( \caV \cup \bar\caV, \caE \cup \bar\caE \bigr).
\ee
Here, $\bar\caV$ is an extra set of vertices (the boundary), and $\bar\caE$ is a set of edges between $\caV$ and $\bar\caV$, i.e.\ $\bar\caE \subset \caV \times \bar\caV$.

The set of link configurations is $\caM_{\caG^{\rm b}} \subset \bbN_0^{\caE \cup \bar\caE}$ and it satisfies the constraints that each site of $\caV$ is touched by an even number of links; there are no constraints at the sites of $\bar\caV$. The set of pairing configurations is $\caP_{\caG^{\rm b}}(\bsm)$; pairings are defined at the sites of $\caV$ only, not at $\bar\caV$. Loops are defined as before, except for {\it open loops} that start and end at the boundary --- they involve exactly two edges touching the boundary (closed loops do not pass by the boundary). Given a wire configuration $\bsw = (\bsm,\bspi) \in \caW_{\caG^{\rm b}}$, we let $\lambda(\bsw)$ denote the number of all loops, counting closed and open loops. The probability of a wire configuration with open boundary conditions is
\be
\bbP_{\caG^{\rm b}}^{\alpha,\bsJ}(\bsw) = \frac1{Z_{\caG^{\rm b}}(\alpha,\bsJ)} \alpha^{\lambda(\bsw)} \biggl( \prod_{e \in \caE \cup \bar\caE} \frac{J_e^{m_e}}{m_e!} \biggr) \exp \biggl\{ -\sum_{x \in \caV} U(n_x(\bsm)) \biggr\}.
\ee
The partition function $Z_{\caG^{\rm b}}(\alpha,\bsJ)$ is defined as expected, so that $\bbP_{\caG^{\rm b}}^{\alpha,\bsJ}$ is a probability distribution.

The main advantage of open boundary conditions is to allow us to introduce the event where a site belongs to long loops, namely that it is connected to the boundary. This is discussed in Section \ref{sec PD}.

\section{Loop correlations and Poisson-Dirichlet distribution}
\label{sec PD}

We now fix the graph to be a large box in $\bbZ^d$ with edges given by nearest-neighbours. We write $\Lambda_L = \{-L,\dots,L\}^d$ for the set of vertices, $\caE_L$ for the set of nearest-neighbours, and $\caG_L = (\Lambda_L,\caE_L)$ for this graph. We also assume that $J_e \equiv J$ is constant. In dimensions $d\geq3$, and if $J$ is large enough, we expect that macroscopic loops are present and that they are described by a Poisson-Dirichlet distribution.

\subsection{Joint distribution of the lengths of macroscopic loops}

These properties can be formulated in various ways. The most direct way is to look at the lengths of the loops in a large box. Recall that the length of a loop is the number of its links. Let $\bigl( \ell_1(\bsw), \ell_2(\bsw), \dots, \ell_k(\bsw) \bigr)$ be the sequence of the lengths of the loops of $\bsw$ in decreasing order, repeated with multiplicities; the number of loops is also random, $k = k(\bsw)$. The ``volume" occupied by the loops is defined as
\be
V(\bsw) = \sum_{j=1}^k \ell_j(\bsw) = \sum_{x\in\Lambda_L} n_x(\bsm).
\ee
We consider the following sequence, which is a random partition of the interval $[0,1]$:
\be
\biggl( \frac{\ell_1(\bsw)}{V(\bsw)}, \frac{\ell_2(\bsw)}{V(\bsw)}, \dots, \frac{\ell_k(\bsw)}{V(\bsw)} \biggr).
\ee
It is rather obvious that the number of microscopic loops (those whose lengths are bounded uniformly in $L$) scales like the volume $|\Lambda_L|$ of the system. Consequently, the tail of the random partition consists of tiny dust occupying a non-vanishing interval. On the other hand, the lengths of the longer loops are expected to be of order of the volume and to be described by a Poisson-Dirichlet distribution. The typical random partition is illustrated in Fig.\ \ref{fig partition}.

\begin{centering}
\bfig
\begin{picture}(0,0)%
\includegraphics{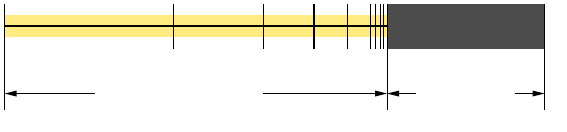}
\end{picture}%
\setlength{\unitlength}{2368sp}%
\begingroup\makeatletter\ifx\SetFigFont\undefined%
\gdef\SetFigFont#1#2#3#4#5{%
  \reset@font\fontsize{#1}{#2pt}%
  \fontfamily{#3}\fontseries{#4}\fontshape{#5}%
  \selectfont}%
\fi\endgroup%
\begin{picture}(7538,1726)(1141,-2965)
\put(2476,-2491){\makebox(0,0)[lb]{\smash{{\SetFigFont{8}{9.6}{\rmdefault}{\mddefault}{\updefault}{\color[rgb]{0,0,0}macroscopic, PD($\theta$)}%
}}}}
\put(6766,-2506){\makebox(0,0)[lb]{\smash{{\SetFigFont{8}{9.6}{\rmdefault}{\mddefault}{\updefault}{\color[rgb]{0,0,0}microscopic}%
}}}}
\put(6226,-2911){\makebox(0,0)[lb]{\smash{{\SetFigFont{8}{9.6}{\rmdefault}{\mddefault}{\updefault}{\color[rgb]{0,0,0}$m$}%
}}}}
\put(1141,-2911){\makebox(0,0)[lb]{\smash{{\SetFigFont{8}{9.6}{\rmdefault}{\mddefault}{\updefault}{\color[rgb]{0,0,0}$0$}%
}}}}
\put(8356,-2911){\makebox(0,0)[lb]{\smash{{\SetFigFont{8}{9.6}{\rmdefault}{\mddefault}{\updefault}{\color[rgb]{0,0,0}$1$}%
}}}}
\end{picture}%
\caption{A typical partition given by loop lengths in dimensions three and higher. The elements in the interval $[0,m]$ are distributed according Poisson-Dirichlet. The elements in the interval $[m,1]$ are due to microscopic loops and they have zero width.}
\label{fig partition}
\efig
\end{centering}

One can formulate the Poisson-Dirichlet conjecture as follows. There exists $m \in [0,1]$ (and $m>0$ when $d\geq3$ and $J$ large enough) such that
\begin{itemize}
\item For every $\eps>0$, we have
\be
\lim_{n\to\infty} \lim_{L\to\infty} \bbP_{\caG_L}^{\alpha,J} \Bigl( \sum_{j=1}^n \frac{\ell_j(\bsw)}{V(\bsw)} \in [m-\eps,m+\eps] \Bigr) = 1.
\ee
\item For every $n\in\bbN$ and as $L \to \infty$, the distribution of the vector $\bigl( \frac{\ell_1(\bsw)}{mV(\bsw)}, \dots, \frac{\ell_n(\bsw)}{mV(\bsw)} \bigr)$ converges to the Poisson-Dirichlet distribution PD($\frac\alpha2$) restricted to the first $n$ elements.
\end{itemize}

Let us recall that Poisson-Dirichlet is a one-parameter family of distributions on partitions of $[0,1]$. It is most easily defined using the random allocation (or ``stick breaking") construction. Namely, let $Y_1,Y_2,\dots$ be i.i.d.\ Beta$(1,\theta)$ random variables (that is, their probability density function is equal to $\theta (1-t)^{\theta-1}$ for $t \in [0,1]$ and is zero otherwise); we construct the sequence
\be
\bigl( Y_1, (1-Y_1) Y_2, (1-Y_1) (1-Y_2) Y_3, \dots).
\ee
It is not hard to check that the sum of these numbers give 1 almost surely. Rearranging the numbers in decreasing order, we get a random partition with Poisson-Dirichlet distribution PD($\theta$). See \cite{Kin,Pit} for more information.

The heuristics for this conjecture is explained in the next subsection; it also contains the calculation of the Poisson-Dirichlet parameter, $\theta = \frac\alpha2$.

\subsection{Heuristics and calculation of the Poisson-Dirichlet parameter}
\label{sec heuristics}

An important property of Poisson-Dirichlet is to be the stationary distribution of split-merge processes (also called coagulation-fragmentation) \cite{Tsi,DMZZ,Ber}. The following heuristic has already been explained in \cite{GUW, GLU, Uel1} for other loop soups; notice that the article \cite{GLU} contains numerical verifications of some of the steps. It is worth sketching the heuristic in some details since it allows to calculate --- non rigorously, but exactly --- the Poisson-Dirichlet parameter. For a fixed link configuration $\bsm$, we introduce a discrete-time Markov process on pairing configurations, i.e.\ on $\caP_{\caG_L}(\bsm)$. Let $T_\bsm(\bspi,\bspi')$ be the probability that, if the system is at $\bspi$ at time $t$, it moves to $\bspi'$ at time $t+1$. Assuming the process to be irreducible (that is, there are possible trajectories reaching all configurations of $\caP_{\Lambda_L}(\bsm)$), a sufficient condition for a measure to be stationary is that it satisfies the {\it detailed balance condition}
\be
\bbP_{\caG_L}^{\alpha,J}(\bsm,\bspi) T_\bsm(\bspi,\bspi') = \bbP_{\caG_L}^{\alpha,J}(\bsm,\bspi') T_\bsm(\bspi',\bspi).
\ee
We only consider changes that involve rewiring two pairs at a single site. This is illustrated in Fig.\ \ref{fig process}. The number of loops changes by at most one. We have
\be
\frac{\bbP_{\caG_L}^{\alpha,J}(\bsm,\bspi')}{\bbP_{\caG_L}^{\alpha,J}(\bsm,\bspi)} = \begin{cases} \alpha & \text{if } \lambda(\bsm,\bspi') = \lambda(\bsm,\bspi)+1, \\ 1 & \text{if } \lambda(\bsm,\bspi') = \lambda(\bsm,\bspi), \\ \alpha^{-1} & \text{if } \lambda(\bsm,\bspi') = \lambda(\bsm,\bspi)-1. \end{cases}
\ee
We need to choose the transition probabilities so that the ratio $\frac{T_\bsm(\bspi,\bspi')}{T_\bsm(\bspi',\bspi)}$ is equal to the above equation when $\bspi$ and $\bspi'$ differ by just one rewiring. There are many possibilities; we can take
\be
T_\bsm(\bspi,\bspi') = \frac{C}{|\Lambda_L|} \, \frac2{(\begin{smallmatrix} 2n_x \\ 2 \end{smallmatrix})} \cdot \begin{cases} \alpha^{1/2} & \text{if } \lambda(\bsm,\bspi') = \lambda(\bsm,\bspi)+1, \\ 1 & \text{if } \lambda(\bsm,\bspi') = \lambda(\bsm,\bspi), \\ \alpha^{-1/2} & \text{if } \lambda(\bsm,\bspi') = \lambda(\bsm,\bspi)-1. \end{cases}
\ee
Here, $C$ is a constant that is small enough so that $\sum_{\bspi' \neq \bspi} T_\bsm(\bspi,\bspi') \leq 1$ (it affects the speed of the process but not its stationary distribution). Notice that $(\begin{smallmatrix} 2n_x \\ 2 \end{smallmatrix})$ is the number of pairs of endpoints at $x$, and there are exactly two choices whose rewiring gives $\bspi'$.

\bfig
\includegraphics[height=40mm]{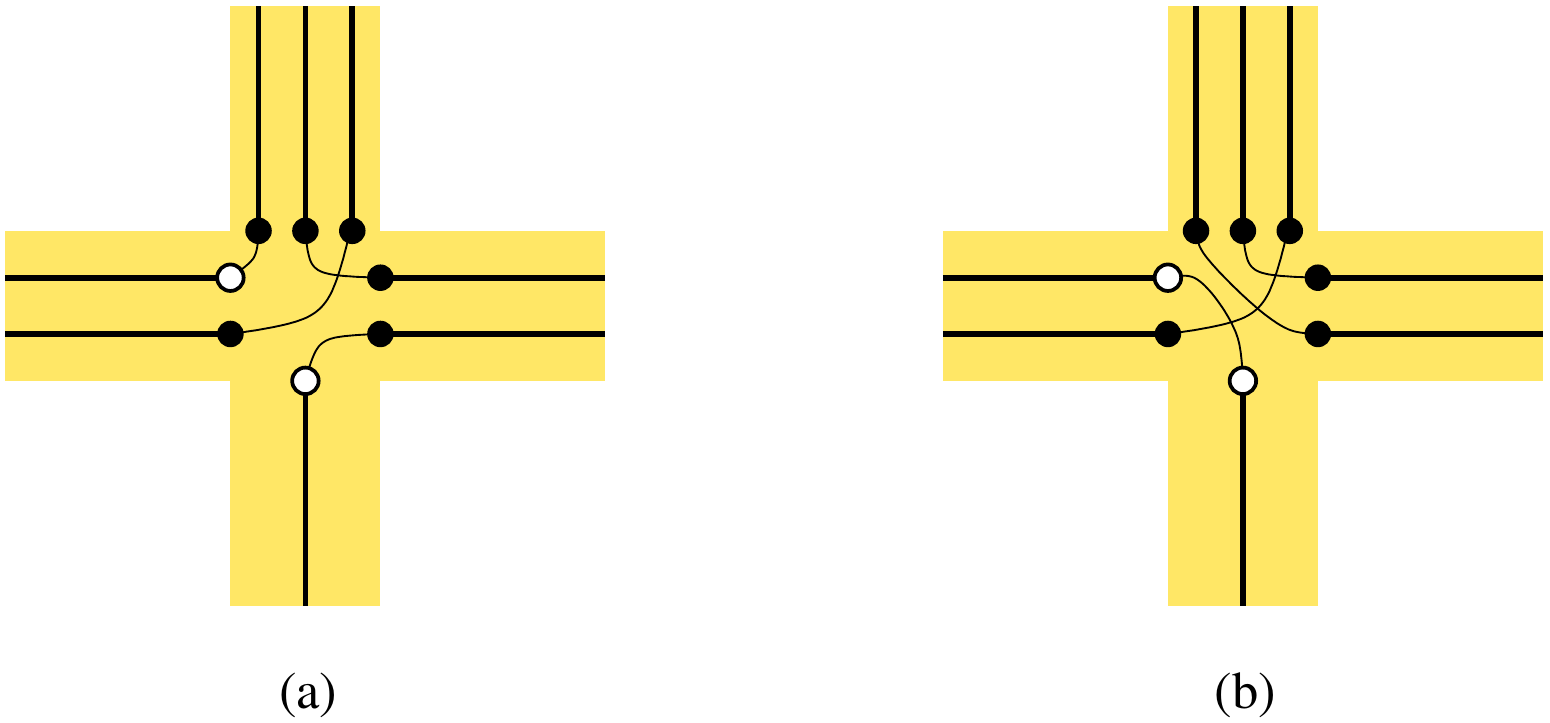}
\caption{Illustration for the Markov process: (a) Two endpoints are selecting at random. (b) The new pairing at this site, if the change is accepted.}
\label{fig process}
\efig

In words, we choose a site uniformly at random, then pick two endpoints uniformly at random, and accept the rewiring with probability $C \alpha^{1/2}, C, C \alpha^{-1/2}$ according to whether the number of loops increases by 1, stays constant, or decreases by 1. It is clear that $T_\bsm$ satisfies the detailed balance condition, and also that the process is irreducible on $\caP_{\Lambda_L}(\bsm)$ for a fixed $\bsm$.

Next, we look at the resulting process on partitions. We get a split-merge process with a priori complicated rates. But we can discard all rewirings that involve microscopic loops, as they have negligible effect in the infinite-volume limit. Much more interesting are changes that affect macroscopic loops. If we select endpoints belonging to different loops, the rewiring always merges them. If we select endpoints belonging to the same loop, the rewiring may split it, or just rearrange it (this is analogous to $0 \leftrightarrow 8$). The essence of the conjecture is that macroscopic loops merge well, and the number of pairs of endpoints that allow two macroscopic loops $\gamma,\gamma'$ to merge is approximately equal to $c \ell_\gamma \ell_{\gamma'}$ for a constant $c$ that is independent of $\gamma,\gamma'$. Further, the number of pairs of endpoints that allow a macroscopic loop $\gamma$ to split is approximately equal to $\frac14 c \ell_\gamma^2$, with the same constant $c$ as before. The factor $\frac14 = \frac12 \cdot \frac12$ is there because pairs within a loop should be counted once, and only half the pairs cause a split and not a rearrangement.

The conclusion of this heuristic is that, as the volume becomes large, the effective split-merge process on partitions behaves like the standard, mean-field process where two partition elements $\eta,\eta'$ merge at rate $\frac{2c}{\sqrt\alpha} \eta \eta'$ and an element $\eta$ splits at rate $\frac{c\sqrt\alpha}2 \eta^2$; moreover, the element is split uniformly. It is known that the Poisson-Dirichlet distribution with parameter $\theta = \frac\alpha2$ is the invariant measure for this process \cite{Tsi,Ber,Uel2} (partial results about uniqueness can be found in \cite{DMZZ}).

This long heuristics was needed in order to identify the correct parameter. This justifies the above conjecture.

\subsection{Poisson-Dirichlet correlations}
\label{sec PD corr}

As we argue below in Section \ref{sec loop corr}, the probability that points belong to the same loop, knowing that they belong to long loops, is given by the probability that random points in the interval $[0,1]$ belong to the same partition element with Poisson-Dirichlet distribution. We collect now the relevant formul\ae.

Let $u_1,\dots,u_k \in [0,1]$ and let $(z_1,z_2,\dots)$ be a partition of $[0,1]$. We denote $X_{u_1,\dots,u_k}(z_1,z_2,\dots)$ the set partition of $\{1,\dots,k\}$ where $i,j$ belong to the same subset if and only if $u_i, u_j$ belong to the same partition element. Further, if $X$ is a set partition of $\{1,\dots,k\}$, let
\be
M_\theta(X; u_1,\dots,u_k) = \bbP_{{\rm PD}(\theta)} \bigl( X_{u_1,\dots,u_k} = X \bigr).
\ee
Finally, let
\be
M_\theta(X) = \bbE_{U_1,\dots,U_k} \bigl[ M_\theta(X; U_1,\dots,U_k) \bigr],
\ee
where the latter expectation is taken over $k$ i.i.d.\ random variables $U_1,\dots,U_k$ with uniform distribution on $[0,1]$

The number $M_\theta(X)$ depends only on the sizes of the partition elements of $X$. If $X = \cup_{i=1}^\ell X_i$ with $|X_i|=n_i$, then, with $Z_i$ denoting the $i$th element of a Poisson-Dirichlet random partition, we have
\be
\label{correl PD}
\begin{split}
M_\theta(X) &= \sumtwo{i_1,\dots,i_\ell \geq 1}{\rm distinct} \int_0^1 \dd u_1 \dots \int_0^1 \dd u_k \; \bbP_{{\rm PD}(\theta)} \bigl( \text{$u_i \in$ $i_j$th element for all $i \in X_j$} \bigr) \\
&= \sumtwo{i_1,\dots,i_\ell \geq 1}{\rm distinct} \bbE_{{\rm PD}(\theta)} \bigl[ Z_{i_1}^{n_1} \dots Z_{i_\ell}^{n_\ell} \bigr] \\
&= \frac{\theta^\ell \Gamma(\theta) \Gamma(n_1) \dots \Gamma(n_\ell)}{\Gamma(\theta+n_1+\dots+n_\ell)}.
\end{split}
\ee
The latter formula seems well-known to experts but it does not appear often in the literature. It is written in \cite{NCSOS} where it is derived using ``supersymmetry" calculations in a loop $O(N)$ model. A calculation within Poisson-Dirichlet can be found in \cite{Uel2}.

In the present article we need the probability that the random set partition is even, that is, all its subsets have an even number of elements. Let $\caX_{2k}^{\rm even}$ denote the set of even partitions of $\{1,\dots,2k\}$, and let
\be
\label{even PD corr}
M_\theta^{\rm even}(2k) = \sum_{X \in \caX_{2k}^{\rm even}} M_\theta(X).
\ee

\begin{proposition}
For all $\theta>0$ and all $k\in\bbN$, we have
\[
M_\theta^{\rm even}(2k) = \frac{\Gamma(2k+1) \Gamma(k+\frac\theta2) \Gamma(\theta)}{\Gamma(2k+\theta) \Gamma(k+1) \Gamma(\frac\theta2)}.
\]
\end{proposition}

In the case $\theta=1$, the formula above reduces to
\be
\label{M1 even}
M_{\theta=1}^{\rm even}(2k) = \frac{(2k-1)!!}{2^k k!}.
\ee

\begin{proof}
We use the following trick:\footnote{We are grateful to Peter M\"orters for the suggestion.} Consider random partitions of $[0,1]$ and a random sequence of signs $(\eps_1,\eps_2,\dots)$ where $\eps_i$ are i.i.d.\ and take values $\pm1$ with probability $\frac12$. Let
\be
\Phi(h) = \bbE_{{\rm PD}(\theta)} \biggl[ \prod_{i\geq1} \cosh(hZ_i) \biggr] = \bbE_{{\rm PD}(\theta) \times (\eps_i)} \bigl[ \e{\sum_{i\geq1} h \eps_i Z_i} \bigr].
\ee
Here, $Z_i$s are the elements of the random partition with PD($\theta$) distribution and $h \in \bbR$. $\Phi(h)$ is an even function and
\be
\begin{split}
\frac{\dd^{2k}}{\dd h^{2k}} \Phi(h) \Big|_{h=0} &= \sum_{i_1,\dots,i_{2k}} \bbE_{{\rm PD}(\theta) \times (\eps_i)} \bigl[ \eps_{i_1} \dots \eps_{i_{2k}} Z_{i_1} \dots Z_{i_{2k}} \bigr] \\
&= \sum_{X \in \caX_{2k}^{\rm even}} M_\theta(X).
\end{split}
\ee
The function $\Phi(h)$ was calculated in \cite[Eq.\ (4.18)]{Uel2}; it is equal to
\be
\Phi(h) = \frac{\Gamma(\theta)}{\Gamma(\frac\theta2)} \sum_{n\geq0} \frac{\Gamma(n+\frac\theta2)}{n! \, \Gamma(2n+\theta)} h^{2n}.
\ee
Differentiating $2k$ times and looking at the coefficient of $h^0$, we get the claim of the proposition.
\end{proof}

\subsection{Loop correlations --- Conjectures}
\label{sec loop corr}
We now formulate the Poisson-Dirichlet conjecture in terms of loop correlations. This is more natural in the context of statistical mechanics, and this is the form that we can prove in a special case.

The idea behind loop connectivity is to consider $k$ points and to look at how these points are connected by the loops. A complication is that many loops may pass by a given site, and also that the same loop may pass many times. We then introduce a label on the pairings. Namely, we assign the labels $1,\dots,n_x(\bsm)$ to the pairs of the pairing $\pi_x$. Given distinct sites $\bsx = (x_1,\dots,x_k)$, pair labels $\bsq = (q_1,\dots,q_k)$, and a set partition $X = \{X_1,\dots,X_\ell\}$ of $\{1,\dots,k\}$, we introduce the event
\be
\begin{split}
\label{def event 1}
E_X(\bsx, \bsq) = \Bigl\{ \bsw &\in \caW_\caG : \text{ $n_{x_i} \geq q_i$ for $i=1,\dots,k$, and $(x_i,q_i), (x_j,q_j)$ belong to} \\
&\text{the same loop iff $i,j$ belong to the same partition element of $X$} \Bigr\}.
\end{split}
\ee
In other words, we look at the partition of the points $(x_j,q_j)_{j=1}^k$ given by the loops, and $E_X(\bsx,\bsq)$ is the event where this partition is equal to $X$. We also use the event $E_\infty(\bsx,\bsq)$, the set of wire configurations where all $(x_i,q_i)$ belong to long loops: With $\ell(x,q)$ denoting the length of the loop passing through the $q$th pair at the site $x \in \Lambda_L$,
\be
E_\infty(\bsx,\bsq) = \bigl\{ \bsw \in \caW_{\caG_L} : \ell(x_i,q_i) \geq \tilde\ell_L \text{ for } i=1,\dots,k \bigr\},
\ee
where the cutoff $\tilde\ell_L$ is chosen so that $\lim_{L\to\infty} \tilde\ell_L = \infty$ and $\lim_{L\to\infty} \tilde\ell_L/L^d = 0$.

\begin{centering}
\bfig
\begin{picture}(0,0)%
\includegraphics{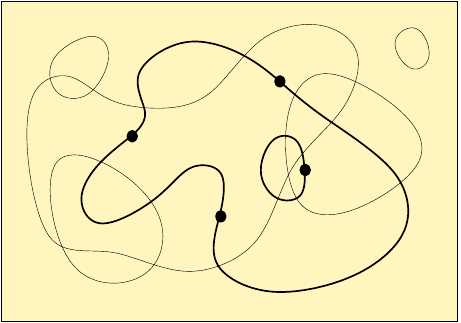}
\end{picture}%
\setlength{\unitlength}{1776sp}%
\begingroup\makeatletter\ifx\SetFigFont\undefined%
\gdef\SetFigFont#1#2#3#4#5{%
  \reset@font\fontsize{#1}{#2pt}%
  \fontfamily{#3}\fontseries{#4}\fontshape{#5}%
  \selectfont}%
\fi\endgroup%
\begin{picture}(8152,5734)(1175,-6078)
\put(5251,-4411){\makebox(0,0)[lb]{\smash{{\SetFigFont{7}{8.4}{\rmdefault}{\mddefault}{\updefault}{\color[rgb]{0,0,0}$x_4$}%
}}}}
\put(6301,-1636){\makebox(0,0)[lb]{\smash{{\SetFigFont{7}{8.4}{\rmdefault}{\mddefault}{\updefault}{\color[rgb]{0,0,0}$x_2$}%
}}}}
\put(6751,-3511){\makebox(0,0)[lb]{\smash{{\SetFigFont{7}{8.4}{\rmdefault}{\mddefault}{\updefault}{\color[rgb]{0,0,0}$x_3$}%
}}}}
\put(3676,-2986){\makebox(0,0)[lb]{\smash{{\SetFigFont{7}{8.4}{\rmdefault}{\mddefault}{\updefault}{\color[rgb]{0,0,0}$x_1$}%
}}}}
\end{picture}%
\caption{Illustration for loop correlations between distant points. In this realisation the set partition is $X = \bigl\{ \{1,2,4\}, \{3\} \bigr\}$.}
\label{fig loop correl}
\efig
\end{centering}

We consider a ``splashing sequence" $\bsx^{(n)} = (x_1^{(n)}, \dots, x_k^{(n)})$, that is, a sequence of sites in $\bbZ^d$ that satisfies
\be
\label{def splashing}
\lim_{n\to\infty} \min_{1 \leq i,j \leq k} \| x_i^{(n)} - x_j^{(n)} \| = \infty.
\ee
The Poisson-Dirichlet conjecture can be formulated as follows: Let $X$ be a set partition without singletons; in the limits $L\to\infty$ then $n\to\infty$, the probability of $E_X(\bsx^{(n)},\bsq)$ involves the probability $\bbP_{\bbZ^d}^{\alpha,J} \bigl( E_\infty(x_i^{(n)},q_i))$ that the $q_i$th pair at $x_i^{(n)}$ belongs to macroscopic loops, and the probability $M_\theta(X)$ that $k$ random numbers placed in a random partition, yields the set partition $X$. More precisely, we expect that for all set partitions $X$ without singletons, we have
\be
\label{conj explained}
\begin{split}
\lim_{n\to\infty} \lim_{L\to\infty} \bbP_{\caG_L}^{\alpha,J} &\bigl( E_X(\bsx^{(n)},\bsq) \bigr) = \lim_{n\to\infty} \lim_{L\to\infty} \bbP_{\caG_L}^{\alpha,J} \bigl( E_X(\bsx^{(n)},\bsq) \cap E_\infty(\bsx^{(n)},\bsq) \bigr) \\
&= \lim_{n\to\infty} \lim_{L\to\infty} \bbP_{\caG_L}^{\alpha,J} \bigl( E_X(\bsx^{(n)},\bsq) \big| E_\infty(\bsx^{(n)},\bsq) \bigr) \; \bbP_{\caG_L}^{\alpha,J} \bigl( E_\infty(\bsx^{(n)},\bsq) \bigr) \\
&= M_{\frac\alpha2}(X) \prod_{i=1}^k \bbP_{\bbZ^d}^{\alpha,J} \bigl( E_\infty(0,q_i) \bigr).
\end{split}
\ee
The first identity should not hold for $X$ with singletons, since the probability that the corresponding points belong to small loops is positive.
Letting $m(d,J) = \sum_{q\geq1} \bbP_{\bbZ^d}^{\alpha,J}\bigl( E_\infty(0,q) \bigr)$, we can also formulate the conjecture as
\be
\label{main conj}
\lim_{n\to\infty} \lim_{L\to\infty} \sum_{\bsq \in \bbN^k} \bbP_{\caG_L}^{\alpha,J} \bigl( E_X(\bsx^{(n)},\bsq) \bigr) =m(d,J)^k \, M_{\frac\alpha2}(X).
\ee
$M_{\frac\alpha2}(X)$ can be found in Eq.\ \eqref{correl PD}.

We now formulate a revised conjecture; it is less appealing but it is closer to what is proved in Theorem \ref{thm PD} below. For all splashing sequences $\bsx^{(n)} = (x_1^{(n)}, \dots, x_k^{(n)})$ and all set partitions $X$ of $\{1,\dots,k\}$ (without singletons), we can repeat the steps of \eqref{conj explained} and we obtain
\be
\lim_{n\to\infty} \lim_{L\to\infty} \bbE_{\caG_L}^{\alpha,J} \biggl[ 1_{E_X(\bsx^{(n)},\bsq)} \prod_{j=1}^k \frac1{n_{x_j^{(n)}}+1} \biggr] = \Bigl( \prod_{j=1}^k \tilde m_{q_j}(d,J) \Bigr) M_{\frac\alpha2}(X),
\ee
with $\tilde m_q(d,J)$ given by
\be
\tilde m_q(d,J) = \lim_{L\to\infty} \bbE_{\caG_L^{\rm b}}\Bigl[ 1_{E_\infty(0,q)} \frac1{n_0+1} \Bigr].
\ee
Letting $\tilde m(d,J) = \sum_{q\geq1} \tilde m_q(d,J)$, the Poisson-Dirichlet conjecture states that for any $k$ and any set partition $X$ of $\{1,\dots,k\}$ without singletons, we have
\be
\label{conj PD corr}
\lim_{n\to\infty} \lim_{L\to\infty} \sum_{\bsq \in \bbN^k} \bbE_{\caG_L}^{\alpha,J} \biggl[ 1_{E_X(\bsx^{(n)},\bsq)} \prod_{j=1}^k \frac1{n_{x_j^{(n)}}+1} \biggr] = \tilde m(d,J)^k M_{\frac\alpha2}(X).
\ee
As in the version \eqref{main conj} of the conjecture, $\tilde m(d,J)$ is related to the density of points in long loops and is model-dependent; $M_{\frac\alpha2}(X)$ is the term that signals the presence of the Poisson-Dirichlet distribution PD($\frac\alpha2$) for the lengths of the long loops.

Replacing $k$ by $2k$, summing over even set partitions, and using \eqref{even PD corr}, we get
\be
\label{partial conj PD}
\sum_{X \in \caX_{2k}^{\rm even}} \lim_{n\to\infty} \lim_{L\to\infty} \sum_{\bsq \in \bbN^{2k}} \bbE_{\caG_L}^{\alpha,J} \biggl[ 1_{E_X(\bsx^{(n)},\bsq)} \prod_{j=1}^k \frac1{n_{x_j^{(n)}}+1} \biggr] = \tilde m(d,J)^{2k} M_{\frac\alpha2}^{\rm even}(2k).
\ee
This is a {\it weaker} conjecture than \eqref{conj PD corr}. We prove it in a specific random wire model, see Theorem \ref{thm PD} in the next section.

\section{Main results --- Long loops and their joint distribution}
\label{sec results}

We can now formulate the main result of this article, which strongly hints towards the presence of the Poisson-Dirichlet distribution in a class of models of random wires. We restrict to the random wire model with loop parameter $\alpha=2$ and potential function defined by
\be
\label{def U2}
\e{-U(n)} = \frac1{n!}.
\ee
The graph is $\caG_L = (\Lambda_L,\caE_L)$ with $\Lambda_L = \{-L,\dots,L\}^d$ and $\caE_L$ is the set of nearest-neighbours. We choose $J_e = J$ for all $e \in \caE_L$. We also consider the graph $\caG_L^{\rm b}$ with boundary conditions; the set of vertices is $\Lambda_L$ with the external boundary
\be
\partial\Lambda_L = \bigl\{ y \in \bbZ^d \setminus \Lambda_L : \exists x \in \Lambda_L \text{ such that } \|x-y\|=1 \bigr\}.
\ee
Edges of $\caG_L^{\rm b}$ are the nearest-neighbours in $\Lambda_L$, and the edges between $\Lambda_L$ and $\partial \Lambda_L$. On this graph some loops are open, with endpoints on $\partial\Lambda_L$. Let $\tilde n_x(\bsw)$ be the random variable for the number of times that open loops pass by the site $x\in\Lambda_L$. In other words, for $\bsw=(\bsm,\bspi)$, we let
\be
\tilde n_x(\bsw) = \tfrac12 \bigl| \tilde L(\bsw) \cap \{ e \in \caE_L : e \ni x \} \bigr|,\label{eq open loops x}
\ee
where $\tilde L(\bsw)$ is the set of links that are connected to the boundary. We then define
\be
\label{def m tilde}
\tilde m(d,J) = \lim_{L\to\infty} \bbE_{\caG_L^{\rm b}}^{2,J} \Bigl[ \frac{\tilde n_0}{n_0+1} \Bigr].
\ee
Ignoring the denominator in the expectation, $\tilde m(d,J)$ gives the average number of pairs at the origin that belong to long loops. The reason why the denominator is present is that $\tilde m(d,J)$ can be written in terms of spin correlations; this allows to establish the following properties, our first main result.

\begin{theorem}
\label{thm long dens}
Let $\alpha=2$ and $U$ defined in Eq.\ \eqref{def U2}. Then the limit $L\to\infty$ of $\tilde m(d,J)$ in Eq.\ \eqref{def m tilde} exists. Further,
\begin{itemize}
\item[(a)] $\tilde m(d,J)$ is nondecreasing with respect to $d$ and $J$.
\item[(b)] $\tilde m(d,J) = 0$ when $J < 2^{-3/2} \log( 1 + \frac1{(2d)^2})$, for arbitrary dimension $d$.
\item[(c)] $\tilde m(d,J) = 0$ when $d=1,2$, for all $J\geq0$.
\item[(d)] For $d\geq3$, there exists $J_{\rm c}(d)<\infty$ such that $\tilde m(d,J) > 0$ if $J > J_{\rm c}(d)$ and $\tilde m(d,J) = 0$ if $J < J_{\rm c}(d)$.
\end{itemize}
\end{theorem}

The proof of this theorem can be found in Section \ref{sec proofs}.

Theorem \ref{thm long dens} establishes that a positive fraction of sites are crossed by long loops when $d\geq3$ and $J$ is large enough. It is remarkable that $\tilde m(d,J)$ can be proved to be monotone nondecreasing in $d$ and in $J$. This property is expected to hold for fairly general random wire models; but the present proof, relying as it does on the equivalent $XY$ spin model and its correlation inequalities, cannot be extended easily.

The claim (b) follows from Proposition \ref{prop basic prop} and it holds for more general $\alpha$ and $U$. When $\alpha = 3,4,5,\dots$, and $U(n)$ is defined by Eq.\ \eqref{def U} with $N=\alpha$, the claim (c) also holds (its proof uses the continuous symmetry of the corresponding spin system).

Next we consider loop correlations between distant points. 
In order to formulate the result, we need to introduce the pressure $p(\alpha,J)$:
\be
p(\alpha,J) = \lim_{L\to\infty} \frac1{|\Lambda_L|} \log Z_{\caG_L}(\alpha,J).
\ee
The infinite-volume limit exists by a standard subadditive argument --- $Z_{\caG_L}$ is submultiplicative, and Proposition \ref{prop basic prop} (a) guarantees that the pressure is finite. It is easy to verify that $p(\alpha,\e{s})$ is convex in $s$, as the second derivative gives the variance of $\sum_e m_e$ and is therefore positive. It follows that $p(\alpha,J)$ is differentiable with respect to $J$ at all points, except possibly for a countable set.

Recall the notion of splashing sequences of sites in \eqref{def splashing}. Our second main result is the  weaker form of the Poisson-Dirichlet conjecture, see Eq.\ \eqref{partial conj PD}.

\begin{theorem}
\label{thm PD}
Let $\alpha=2$, $U$ defined in Eq.\ \eqref{def U2}, and $\tilde m(d,J)$ defined in Eq.\ \eqref{def m tilde}. We assume that $J$ is such that the pressure $p(2,J)$ is differentiable. Then for all $k \in \bbN$ and all splashing sequences of $2k$ sites, we have
\[
\sum_{X \in \caX_{2k}^{\rm even}} \lim_{n\to\infty} \lim_{L\to\infty} \sum_{\bsq \in \bbN^{2k}} \bbE_{\caG_L}^{2,J} \biggl[ 1_{E_X(\bsx^{(n)},\bsq)} \; \prod_{j=1}^{2k} \frac1{n_{x_j^{(n)}} + 1} \biggr] = \bigl( \tilde m(d,J) \bigr)^{2k} M_1^{\rm even}(2k).
\]
\end{theorem}

The proof of this theorem uses the connections to the classical XY model; it can be found in Section \ref{sec proofs}.

Theorem \ref{thm PD} gives a lot of information on the structure of long loops: They are present when $\tilde m(d,J)>0$; arbitrary sites have positive probability to belong to them; multiple long loops occur with positive probability. An important aspect of Theorem \ref{thm PD} is that it holds for all $k$ with the same constant $\tilde m(d,J)$. This is compatible with the Poisson-Dirichlet distribution PD($\theta$) with $\theta=1$; this is {\it incompatible} with PD($\theta$) with $\theta\neq1$ and with most other distributions on partitions. Theorem \ref{thm PD} is then a good step forward towards proving that the correlations due to long loops are given by PD(1).

\section{Random wire representation of classical $O(N)$ spin systems}
\label{sec wire O(N)}

We show now that the random wire model can be derived as a representation of classical $O(N)$ spin systems. In fact, the case $N=1$ is close to the random current representation of the Ising model \cite{GHS,Aiz,FV}. The general case $N\in\bbN$ can be seen as a reformulation of the Brydges-Fr\"ohlich-Spencer loop model \cite{BFS,FFS}; explicit relations between BFS loops and wire configurations can be found in \cite{Ben}.

We consider an arbitrary finite graph $\caG = (\caV,\caE)$. Let $\bsJ = (J_e)_{e\in\caE}$ be fixed parameters. We denote $\bsvarphi \in (\bbS^{N-1})^\caV$ the spin configurations. The hamiltonian of the $O(N)$ spin system is defined as
\be
\label{def ham spin}
H_\caG^\bsJ(\bsvarphi) = -\sum_{e = \{x,y\} \in \caE} 2J_e \; \varphi_x \cdot \varphi_y,
\ee
where $\varphi_x \cdot \varphi_y$ denotes the usual inner product of two $N$-component vectors.
The partition function is
\be
Z_\caG^{\rm spin}(\bsJ) = \biggl( \prod_{x\in\caV} \int_{\bbS^{N-1}} \dd\varphi_x \biggr) \e{-H_\caG^\bsJ(\bsvarphi)}.
\ee
Here, $\dd\varphi_x$ denotes the uniform probability measure on $\bbS^{N-1}$, that is, $\int_{\bbS^{N-1}} \dd\varphi_x = 1$. The relevant Gibbs state can be defined as the linear functional $\langle \cdot \rangle_\caG^\bsJ$ on functions $(\bbS^{N-1})^\caV \to \bbR$, that assigns the value
\be
\langle f \rangle_\caG^\bsJ = \frac1{Z_\caG^{\rm spin}(\bsJ)} \biggl( \prod_{x\in\caV} \int_{\bbS^{N-1}} \dd\varphi_x \biggr) f \bigl( (\varphi_x)_{x\in\caV} \bigr) \e{-H_\caG^\bsJ(\bsvarphi)}.
\ee

We introduce a special class of spin correlation functions that have special relevance to loop models. Let $k\in\bbN$ and $x_1, \dots, x_{2k} \in \caV$ be distinct sites. We assume that $N\geq2$ and we write $\varphi_x^{(i)}$ for the $i$th component of the vector $\varphi_x \in \bbR^N$. The corresponding correlation function is
\be
\label{def main corr}
\langle \varphi_{x_1}^{(1)} \varphi_{x_1}^{(2)} \dots \varphi_{x_{2k}}^{(1)} \varphi_{x_{2k}}^{(2)} \rangle_\caG^\bsJ = \frac{Z_\caG^{\rm spin}(\bsJ; x_1,\dots,x_{2k})}{Z_\caG^{\rm spin}(\bsJ)}
\ee
with
\be
Z_\caG^{\rm spin}(\bsJ; x_1,\dots,x_{2k}) = \biggl( \prod_{x\in\caV} \int_{\bbS^{N-1}} \dd\varphi_x \biggr)  \varphi_{x_1}^{(1)} \varphi_{x_1}^{(2)} \dots \varphi_{x_{2k}}^{(1)} \varphi_{x_{2k}}^{(2)} \e{-H_\caG^\bsJ(\bsvarphi)}.
\ee

Let us define the potential function $U$ of the random wire model by the equation
\be
\label{def U}
\e{-U(n)} = \frac{\Gamma(\frac N2)}{\Gamma(n+\frac N2)}.
\ee
We then have a relation between the $O(N)$ spin system and the random wire model with $\alpha=N$. Recall the event $E_X(\bsx,\bsq)$ defined in \eqref{def event 1}.

\begin{proposition}
\label{prop repr O(N)}
Let $U(n)$ defined by \eqref{def U}. Then
\begin{itemize}
\item[(a)] If $N \in \bbN$, we have
\[
Z_\caG^{\rm spin}(\bsJ) = Z_\caG(N,\bsJ).
\]
\item[(b)] If $N=2,3,4,\dots$, we have
\[
\langle \varphi_{x_1}^{(1)} \varphi_{x_1}^{(2)} \dots \varphi_{x_{2k}}^{(1)} \varphi_{x_{2k}}^{(2)} \rangle_\caG^\bsJ = \sum_{X \in \caX^{\rm even}_{2k}} \Bigl( \frac2N \Bigr)^{|X|} \frac1{2^{2k}} \sum_{\bsq \in \bbN^{2k}} \bbE_\caG^{N,\bsJ} \biggl[ 1_{E_X(\bsx,\bsq)} \prod_{j=1}^{2k} \frac1{n_{x_j}+\frac N2} \biggr].
\]
\end{itemize}
\end{proposition}

Recall that $\caX^{\rm even}_{2k}$ is the set of even set partitions of $\{1,\dots,2k\}$.
It is possible to consider other correlation functions, for instance $\langle \varphi_x^{(1)} \varphi_y^{(1)} \rangle_\caG^\bsJ$. They can be written as ratios of loop partition functions, with the numerator involving ``open" configurations of links where $2n_x$ and $2n_y$ are odd. See \cite{Aiz,FV} for the Ising random currents and \cite{BFS,LT} for the related loop model for $O(N)$ spin systems. But these correlations do not have a direct probability meaning and we ignore them in this article.

In the case $N=1$ we have $\e{-U(n)} = 2^n/(2n-1)!!$; the denominator is equal to the number of pairings of $2n$ elements.

\begin{proof}
Let $\bsx = (x_1,\dots,x_{2k})$.
We get an expansion for $Z_\caG^{\rm spin}(\bsJ; \bsx)$ that also applies to the case $k=0$, i.e.\ $\bsx = \emptyset$. Let $\caM_\caG(\bsx) \subset \bbN_0^\caE$ be the set of link configurations with odd numbers of links touching $x_1,\dots,x_{2k}$, and even numbers touching all other sites. We write
\be
\exp\biggl\{ \sum_{e = \{x,y\} \in \caE} 2J_e \; \varphi_x \cdot \varphi_y \biggr\} = \prod_{e = \{x,y\} \in \caE} \; \prod_{i=1}^N \e{2J_e \, \varphi_x^{(i)} \varphi_y^{(i)}},
\ee
and we expand the exponential in Taylor series. Recalling the definition \eqref{def local occupancy}, we find
\bm
\label{genesis}
Z_\caG^{\rm spin}(\bsJ; \bsx) = \sum_{\bsm^{(1)} \in \caM_\caG(\bsx)} \sum_{\bsm^{(2)} \in \caM_\caG(\bsx)} \sum_{\bsm^{(3)} \in \caM_\caG} \dots \sum_{\bsm^{(N)} \in \caM_\caG} \biggl( \prod_{e \in \caE} \frac{(2J_e)^{m_e}}{m_e^{(1)}! \dots m_e^{(N)}!} \biggr) \\
\biggl( \prod_{x\in\caV} \int_{\bbS^{N-1}} \dd\varphi_x \biggr) \biggl( \prod_{x \in \caV \setminus \bsx} \bigl( \varphi_x^{(1)} \bigr)^{2n_x^{(1)}} \dots \bigl( \varphi_x^{(N)} \bigr)^{2n_x^{(N)}} \biggr) \\
\prod_{j=1}^{2k} \bigl( \varphi_{x_j}^{(1)} \bigr)^{2n_{x_j}^{(1)}+1} \bigl( \varphi_{x_j}^{(2)} \bigr)^{2n_{x_j}^{(2)}+1} \bigl( \varphi_{x_j}^{(3)} \bigr)^{2n_{x_j}^{(3)}} \dots \bigl( \varphi_{x_j}^{(N)} \bigr)^{2n_{x_j}^{(N)}}.
\end{multline}
We set $m_e = \sum_{i=1}^N m_e^{(i)}$. We restricted the link configurations to the sets $\caM_\caG(\bsx)$ or $\caM_\caG$ since the angular integrals vanish otherwise by symmetry.
Recall that $\dd\varphi$ denote the normalised uniform measure on $\bbS^{N-1}$; we now use that
\be
\label{eq int spins}
\int_{\bbS^{N-1}} \bigl( \varphi^{(1)} \bigr)^{2n^{(1)}} \dots \bigl( \varphi^{(N)} \bigr)^{2n^{(N)}} \dd\varphi = \frac{\Gamma(\frac N2)}{2^n \Gamma(n+\frac N2)} \prod_{i=1}^N (2n^{(i)}-1)!!,
\ee
where $n = n^{(1)} + \dots + n^{(N)}$ and with the convention that $(-1)!! = 1$.
We rewrite the expansion by first summing over $\bsm \in \caM_\caG$. We then sum over $(m_e^{(1)})$, \dots, $(m_e^{(N)})$ such that $m_e^{(1)} + \dots + m_e^{(N)} = m_e$ for all $e \in \caE$. We get
\bm
Z_\caG^{\rm spin}(\bsJ; \bsx) = 2^{-2k} \sum_{\bsm \in \caM_\caG} \biggl( \prod_{e\in\caE} \frac{J_e^{m_e}}{m_e!} \biggr) \sumthree{\bsm^{(1)}, \bsm^{(2)} \in \caM_\caG(\bsx)}{\bsm^{(3)}, \dots, \bsm^{(N)} \in \caM_\caG}{\bsm^{(1)} + \dots + \bsm^{(N)} = \bsm} \prod_{e \in \caE} \frac{m_e!}{m_e^{(1)}! \dots m_e^{(N)}!} \\
\prod_{x \in \caV \setminus \bsx} \biggl( \frac{\Gamma(\frac N2)}{\Gamma(n_x + \frac N2)} \prod_{i=1}^N (2n_x^{(i)}-1)!! \biggr) %\\
\prod_{j=1}^{2k} \biggl( \frac{\Gamma(\frac N2)}{\Gamma(n_{x_j}+1+\frac N2)} (2n_{x_j}^{(1)})!! \; (2n_{x_j}^{(2)})!! \prod_{i=3}^N (2n_{x_j}^{(i)}-1)!! \biggr).
\end{multline}

We now replace the sums over $(m_e^{(i)})$ by a sum over $N$ possible ``colours" for each link, subject to the constraint that each site is intersected by an even number of links of each colour --- except for the sites $x_1,\dots,x_{2k}$, which are intersected by an odd number of 1-links and 2-links, and an even number of links of other colours. Further, we replace $(2n_x^{(i)}-1)!!$ by a sum over pairings of the $i$-links that intersect the site $x$. As for the sites $x_1,\dots,x_{2k}$, we sum over pairings such that a 1-link is paired with a 2-link, and all other pairs are between links of same colour. The number of choices for 1-links is $2n_{x_j}^{(1)}$ times $(2 n_{x_j}^{(1)}-2)!!$, the number of pairings of the remaining $2n_{x_j}^{(1)}-1$ points. The number of such pairings is then
\be
(2n_{x_j}^{(1)})!! \; (2n_{x_j}^{(2)})!! \prod_{i=3}^N (2n_{x_j}^{(i)}-1)!!.
\ee
Let $C(\bsw,\bsx)$ be the set of colour configurations that are compatible with the wire configuration $\bsw = (\bsm,\bspi)$ and the sites $\bsx$. We obtain
\be
\label{Hello Costanza!}
Z_\caG^{\rm spin}(\bsJ;\bsx) =  \sum_{\bsm \in \caM_\caG} \biggl( \prod_{e\in\caE} \frac{J_e^{m_e}}{m_e!} \biggr) \sum_{\bspi \in \caP_\caG(\bsm)} \biggl( \prod_{x\in\caV} \frac{\Gamma(\frac N2)}{\Gamma(n_x + \frac N2)} \biggr) \biggl( \prod_{j=1}^{2k} \frac1{2n_{x_j} + N} \biggr) |C(\bsw,\bsx)|.
\ee
The number of colours for a given $\bsm,\bspi,\bsx$ can be expressed in terms of loops. If $k=0$, i.e.\ without the complications due to $\bsx$, the constraint is that the links must have the same colour if they belong to the same loop. Then $|C(\bsw,\emptyset)| = N^{\lambda(\bsw)}$ and the claim (a) of the theorem is proved.

For $k\geq1$ the constraints from $\bsx$ are that there must be loops crossing these sites, whose colours change from 1 to 2 (or 2 to 1). We first sum over the pairs $q_1, \dots, q_{2k}$ where the changes occur. Then the wire configuration $\bsw$ must belong to a set $E_X(\bsx,\bsq)$ defined in \eqref{def event 1} for some even partition $X$ of $\{1,\dots,2k\}$. In that case there are $N$ possible colours for ordinary loops, and 2 colours for loops with changes $1 \leftrightarrow 2$. The number of colourings is then $N^{\lambda(\bsw)} (\frac2N)^{|X|}$ with $|X|$ the number of partition elements. Thus
\be
|C(\bsw,\bsx)| = \sum_\bsq \sum_{X \in \caX_{2k}^{\rm even}} N^{\lambda(\bsw)} \bigl( \tfrac2N \bigr)^{|X|} 1_{E_X(\bsx,\bsq)}(\bsw).
\ee
This gives the claim (b) of the theorem.
\end{proof}

We now consider the graph $\caG^{\rm b}$ with boundary. The hamiltonian is
\be
\label{def ham spin bc}
H_{\caG^{\rm b}}^{J,\bsone}(\bsvarphi) = -2J \sum_{\{x,y\} \in \caE} \varphi_x \cdot \varphi_y - \sqrt2 J \sum_{\{x,y\} \in \bar\caE} \varphi_x \cdot \bsone
\ee
where $\bsone$ is the $N$-component vector $(1,\dots,1)$. The partition function is
\be
Z_{\caG^{\rm b}}^{{\rm spin},\bsone}(J) = \biggl( \prod_{x \in \caV} \int_{\bbS^{N-1}} \dd\varphi_x \biggr) \e{-H_{\caG^{\rm b}}^{J,\bsone}(\bsvarphi)}.
\ee
We write $\langle \cdot \rangle_{\caG^{\rm b}}^{J,\bsone}$ for the Gibbs state with boundary condition $\bsone$.

\begin{proposition}
\label{prop spin open loops}
We have
\begin{itemize}
\item[(a)] $\displaystyle Z_{\caG^{\rm b}}^{{\rm spin},\bsone}(J) = Z_{\caG^{\rm b}}(N,J)$.
\item[(b)] If $N\geq2$, $\displaystyle \langle \varphi_x^{(1)} \varphi_x^{(2)} \rangle_{\caG^{\rm b}}^{J,\bsone} = \tfrac1N \, \bbE_{\caG^{\rm b}}^{N,J} \biggl[ \frac{\tilde n_x}{n_x+\frac N2} \biggr]$.
\end{itemize}
\end{proposition}

\begin{proof}
The claim (a) can be proved as Proposition \ref{prop repr O(N)} (a). The relation between $(m_e)$ and $(n_x)$ is
\be
\sum_{x\in\caV} n_x = \sum_{e \in \caE} m_e + \tfrac12 \sum_{e \in \bar\caE} m_e.
\ee
Thus the factor $2^{-\sum n_x}$ from \eqref{eq int spins} kills the factors $2$ and $\sqrt2$ in front of the coupling parameters. The number of colours for a wire configuration $\bsw \in \caW_{\caG^{\rm b}}$ is equal to $N^{\lambda(\bsw)}$, where $\lambda(\bsw)$ is the total number of closed and open loops.

The claim (b) is also similar to Proposition \ref{prop repr O(N)} (b). Let
\be
Z^{{\rm spin}, \bsone}_{\caG^{\rm b}}(J; \,x) =  \biggl( \prod_{y \in \caV} \int_{\bbS^{N-1}} \dd\varphi_y \biggr)\varphi^{(1)}_x \varphi^{(2)}_x \e{-H_{\caG^{\rm b}}^{J,\bsone}}.
\ee
Proceeding as before, we get the analogue of \eqref{Hello Costanza!}. With $\caM_{\caG^{\rm b}}(x)$ the set of link configurations with an odd number of links touching $x$ and an even number touching all other sites of $\caV$, we have
\be
Z^{{\rm spin}, \bsone}_{\caG^{\rm b}} (J;x) = \sum_{\bsm \in \caM_{\caG^{\rm b}}} \biggl(\prod_{e \in \caE \cup \bar\caE} \frac{J^{m_e}}{m_e!} \biggr) \sum_{\bspi \in \caP_{\caG^{\rm b}}(\bsm)} \biggl(\prod_{y\in\caV} \frac{\Gamma(\frac{N}{2})}{\Gamma(n_y + \frac{N}{2})} \biggr) \\
\frac{1}{2n_x + N} \, |C(\bsw,x)|.
\ee
With $\tilde n_x$ the number of pairs at $x$ that belong to open loops, the number of colours is
\be
|C(\bsw,x)| = 2 \tilde n_x N^{\lambda(\bsw)-1}.
\ee
We get Proposition \ref{prop spin open loops} (b).
\end{proof}

\section{Correlations of $O(2)$ spin systems}
\label{sec proofs}

We now calculate the correlation function \eqref{def main corr}. The idea is to use Pfister's theorem on the characterisation of translation-invariant Gibbs states for the $O(2)$ spin model \cite{Pfi2}. In this section the graph is $\caG_L^{\rm b}$, that is, a box in $\bbZ^d$ with boundary conditions.

It is convenient to introduce the angles $\bsphi = (\phi_x)_{x \in \Lambda} \in [0,2\pi)^{\Lambda_L}$ such that $\varphi_x = (\cos\phi_x, \sin\phi_x)$. In these variables, the hamiltonians \eqref{def ham spin} and \eqref{def ham spin bc} are
\be
\begin{split}
&H_{\caG_L}^J(\bsphi) = -2J \sum_{\{x,y\} \in \caE_L} \cos(\phi_x-\phi_y), \\
&H_{\caG_L^{\rm b}}^{J,\overline\bsphi}(\bsphi) = -2J \sum_{\{x,y\} \in \caE_L} \cos(\phi_x-\phi_y) - 2J \sumtwo{x \in \Lambda_L, y \in \partial\Lambda_L}{\|x-y\|=1} \cos(\phi_x - \overline\phi_y).
\end{split}
\ee
The boundary condition $\bsone$ with variables $\{\varphi_x\}$ corresponds to $\overline\bsphi = (\frac\pi4)_{x \in \partial\Lambda_L}$. The corresponding Gibbs state for free boundary conditions is the linear functional that assigns the value
\be
\langle f \rangle_{\Lambda_L}^J = \frac1{Z_{\caG_L}^{\rm spin}(J)} \Bigl( \prod_{x\in\Lambda_L} \frac1{2\pi} \int_0^{2\pi} \dd\phi_x \Bigr) f(\bsphi) \e{-H_{\caG_L}^J(\bsphi)}
\ee
to a function $f : [0,2\pi)^{\Lambda_L} \to \bbR$. For boundary conditions $\overline\bsphi$, the definition of $\langle\cdot\rangle_{\caG_L^{\rm b}}^{J,\overline\bsphi}$ is the same but with hamiltonian $H_{\caG_L^{\rm b}}^{J,\overline\bsphi}$.

We have $\varphi_x^{(1)} \varphi_x^{(2)} = \cos\phi_x \sin\phi_x = \frac12 \sin(2\phi_x)$. We rotate all spins by $-\frac\pi4$ so as to get the more traditional $\overline\bsphi \equiv 0$ boundary conditions. Since $\sin(2(\phi_x+\frac\pi4)) = \cos(2\phi_x)$, we obtain
\be
\label{phivarphi}
\langle \varphi_x^{(1)} \varphi_x^{(2)} \rangle_{\caG_L^{\rm b}}^{J,\bsone} = \tfrac12 \langle \cos(2\phi_x) \rangle_{\caG_L^{\rm b}}^{J,\overline\bsphi \equiv 0}.
\ee

We are going to use a major result of Pfister about the set of extremal states of the classical XY model \cite{Pfi2}. In order to state this result, let $\langle\cdot\rangle_{\bbZ^d}^J$ and $\langle\cdot\rangle_{\bbZ^d}^{J,0}$ denote the infinite-volume Gibbs states
\be
\langle\cdot\rangle_{\bbZ^d}^J = \lim_{L\to\infty} \langle \cdot \rangle_{\caG_L^{\rm b}}^J, \qquad
\langle\cdot\rangle_{\bbZ^d}^{J,0} = \lim_{L\to\infty} \langle \cdot \rangle_{\caG_L^{\rm b}}^{J,\overline\phi \equiv 0}.
\ee
Existence of the limits $L\to\infty$ follows from Ginibre's inequalities \cite{Gin} with standard arguments, see e.g.\ \cite{FV}. As a matter of fact the infinite-volume limits can be taken along any ``van Hove sequence" of increasing domains, which implies in particular that the limiting states are translation-invariant.
Then Pfister's theorem states that the limiting symmetric Gibbs state $\langle\cdot\rangle_{\bbZ^d}^J$ is equal to the following convex combination of extremal states:
\be
\label{symm state}
\langle\cdot\rangle_{\bbZ^d}^J = \tfrac1{2\pi} \int_0^{2\pi} \langle\cdot\rangle_{\bbZ^d}^{J,\overline\bsphi \equiv \psi} \dd\psi.
\ee
Notice that the state $\langle\cdot\rangle_{\bbZ^d}^{J,\psi}$ is obtained from $\langle\cdot\rangle_{\bbZ^d}^{J,0}$ by a global spin rotation of angle $-\psi \in [0,2\pi)$. The above decomposition holds for all $J$ such that the pressure $p(2,J)$ is differentiable. 

We can now prove Theorems \ref{thm long dens} and \ref{thm PD}.

\begin{proof}[Proof of Theorem \ref{thm long dens}] 
From its definition \eqref{def m tilde}, Proposition \ref{prop spin open loops} (b), and Eq.\ \eqref{phivarphi}, we have that
\be
\tilde m(d,J) = \langle \cos(2\phi_0) \rangle_{\bbZ^d}^{J,0}.
\ee

The monotonicity properties of Theorem \ref{thm long dens} (a) follow from standard arguments based on Ginibre's inequalities, see \cite{FV}.

For Theorem \ref{thm long dens} (b), we use Proposition \ref{prop basic prop} (b) --- more precisely, we use a straightforward extension to the case of open boundary conditions. We take $\bar\alpha = \sqrt2$ and $C=2$. The number of random walks of length $k$ and with fixed initial point is equal to $(2d)^k$. This immediately gives the result.

The absence of long loops when $d=1$ is an easy exercise, and when $d=2$ it follows from the works of Pfister \cite{Pfi1} and Ioffe, Shlosman, and Velenik \cite{ISV}; see \cite[Theorem 9.2]{FV} for a clear exposition. Their result is that the infinite-volume Gibbs state is invariant under spin rotations, so $\tilde m(2,J) = \langle \cos(2\phi_0) \rangle_{\bbZ^2}^{J,0} = 0$. This proves (c).

For (d), it can be shown that $\bigl\langle \cos\phi_0 \bigr \rangle_{\bbZ^d}^{J,0} > 0$ implies that $\bigl\langle \cos(2\phi_0) \bigr \rangle_{\bbZ^d}^{J,0} > 0$, see \cite[Corollary 3.6]{Pfi1}. We now use the fundamental result of Fr\"ohlich, Simon, Spencer about the occurrence of long-range order in $O(N)$. Let $\langle\cdot\rangle_{\bbZ^d}^{J,{\rm per}}$ denote the infinite-volume Gibbs state obtained as the limit $L\to\infty$ of the state $\langle\cdot\rangle_{\caG_L}^{J,{\rm per}}$ with even $L$ and periodic boundary conditions. The claim \cite[Theorem 3.1]{FSS} is that
\be
\label{FSS}
\lim_{\|x\|\to\infty} \bigl\langle \cos\phi_0 \cos\phi_x \bigr\rangle_{\bbZ^d}^{J,{\rm per}} = c(d,J),
\ee
with $c(d,J)>0$ for $d\geq3$ and $J$ large enough --- the theorem actually holds for all $N \in \bbN$, not only $N=2$. Since the state $\bigl\langle \cdot \bigr\rangle_{\caG_L}^{J,{\rm per}}$ is translation and rotation invariant, the infinite-volume limit is equal to the state in \eqref{symm state}. Then
\be
\begin{split}
\lim_{\|x\|\to\infty} \bigl\langle \cos\phi_0 \cos\phi_x \bigr\rangle_{\bbZ^d}^{J,{\rm per}} &= \lim_{\|x\|\to\infty} \frac1{2\pi} \int_0^{2\pi} \bigl\langle \cos\phi_0 \cos\phi_x \bigr\rangle_{\bbZ^d}^{J,\psi} \dd\psi \\
&= \frac1{2\pi} \int_0^{2\pi} \Bigl( \bigl\langle \cos\phi_0 \bigr\rangle_{\bbZ^d}^{J,\psi} \Bigr)^2 \dd\psi \\
&= \tfrac12 \bigl( \langle \cos\phi_0 \rangle_{\bbZ^d}^{J,0} \bigr)^2.
\end{split}
\ee
We used $\langle \cos\phi_0 \rangle_{\bbZ^d}^{J,\psi} = \langle \cos(\phi_0+\psi) \rangle_{\bbZ^d}^{J,0} = \cos\psi \; \langle \cos\phi_0 \rangle_{\bbZ^d}^{J,0}$ and we integrated the angular integral. It follows that $\langle \cos\phi_0 \rangle_{\bbZ^d}^{J,0} = \sqrt{2c(d,J)}$ is positive for $J$ large enough, and so is $\tilde m(d,J)$.

Notice that the extremal state decomposition \eqref{symm state} is only proved for almost all $J$; but using the claim (a) about monotonicity in $J$, we get the existence of $J_{\rm c}$ as stated in (d).
\end{proof}

\begin{proof}[Proof of Theorem \ref{thm PD}] 
By Proposition \ref{prop repr O(N)} (b), the left side of the equation of Theorem \ref{thm PD} is equal to the limits $L\to\infty$ then $n\to\infty$ of the correlation function $2^{2k} \bigl\langle \varphi_{x_1^{(n)}}^{(1)} \varphi_{x_1^{(n)}}^{(2)} \dots \varphi_{x_{2k}^{(n)}}^{(1)} \varphi_{x_{2k}^{(n)}}^{(2)} \bigr\rangle_{\caG_L}^J $.

We use Pfister's theorem \eqref{symm state} and the fact that extremal states are clustering; we get
\be
\begin{split}
2^{2k} \bigl\langle \varphi_{x_1^{(n)}}^{(1)} \varphi_{x_1^{(n)}}^{(2)} \dots \varphi_{x_{2k}^{(n)}}^{(1)} \varphi_{x_{2k}^{(n)}}^{(2)} \bigr\rangle_{\caG_L}^J &= \; \bigl\langle \sin(2\phi_{x_1^{(n)}}) \dots \sin(2\phi_{x_{2k}^{(n)}}) \bigr\rangle_{\caG_L}^J \\
& \substack{L\to\infty \\ \longrightarrow \\ \phantom{L}} \frac1{2\pi} \int_0^{2\pi} \bigl\langle \sin(2\phi_{x_1^{(n)}}) \dots \sin(2\phi_{x_{2k}^{(n)}}) \bigr\rangle_{\bbZ^d}^{J,\psi} \dd\psi\\
& \substack{n\to\infty \\ \longrightarrow \\ \phantom{L}} \frac1{2\pi} \int_0^{2\pi} \Bigl( \bigl\langle \sin(2\phi_0) \bigr\rangle_{\bbZ^d}^{J,\psi} \Bigr)^{2k} \dd\psi.
\end{split}
\ee
The expectation in the rotated Gibbs state can be expressed in term of $\tilde m(d,J)$, namely,
\be
\begin{split}
\bigl\langle \sin(2\phi_0) \bigr\rangle_{\bbZ^d}^{J,\psi} &= \bigl\langle \sin(2\phi_0 + 2\psi) \bigr\rangle_{\bbZ^d}^{J,0} \\
&= \cos(2\psi) \bigl\langle \sin(2\phi_0) \bigr\rangle_{\bbZ^d}^{J,0} + \sin(2\psi) \bigl\langle \cos(2\phi_0) \bigr\rangle_{\bbZ^d}^{J,0}.
\end{split}
\ee
We have $\bigl\langle \sin(2\phi_0) \bigr\rangle_{\bbZ^d}^{J,0} = 0$ by symmetry $\phi_x \mapsto -\phi_x$. We recognise $\tilde m(d,J)$ in the last term. We obtain
\be
\begin{split}
\lim_{n\to\infty} \lim_{L\to\infty} 2^{2k} \bigl\langle\varphi_{x_1^{(n)}}^{(1)} \varphi_{x_1^{(n)}}^{(2)} \dots \varphi_{x_{2k}^{(n)}}^{(1)} \varphi_{x_{2k}^{(n)}}^{(2)} \bigr\rangle_{\caG_L}^J &= \tilde m(d,J)^{2k} \frac1{2\pi} \int_0^{2\pi} \sin^{2k}(2\psi) \, \dd\psi \\
&= \tilde m(d,J)^{2k} \frac{(2k-1)!!}{2^k k!}.
\end{split}
\ee
This is precisely the formula \eqref{M1 even} for $M^{\rm even}_{\theta=1}(2k)$. This completes the proof of Theorem \ref{thm PD}.
\end{proof}

\bigskip
\noindent
{\bf Acknowledgments:} We are grateful to Jakob Bj\"ornberg, J\"urg Fr\"ohlich, Peter M\"orters, Charles-\'Edouard Pfister, Vedran Sohinger, Akinori Tanaka, and Yvan Velenik, for useful discussions. We thank the referee for valuable comments. CB is supported by the Leverhulme Trust Research Project Grant  RPG-2017-228.

\newpage

\renewcommand{\refname}{\small References}
\bibliographystyle{symposium}

\end{document}